\documentclass[10pt]{article}

\usepackage{amsfonts}
\usepackage{amsmath}
\usepackage{bbold}
\usepackage{amssymb}
\usepackage{amsfonts}
\usepackage{latexsym}
\usepackage{color}
\usepackage{float}
\usepackage{tikz}
\usepackage{graphicx,graphics,epstopdf}
\usepackage{amscd}
\usepackage{subfigure}

\catcode `\@=11 \@addtoreset{equation}{section}

\catcode `\@=12

\newcommand{\be}{\begin{equation}}
\newcommand{\en}{\end{equation}}
\newcommand{\bea}{\begin{eqnarray}}
\newcommand{\ena}{\end{eqnarray}}
\newcommand{\beano}{\begin{eqnarray*}}
\newcommand{\enano}{\end{eqnarray*}}
\newcommand{\bee}{\begin{enumerate}}
\newcommand{\ene}{\end{enumerate}}

\newcommand{\mb}{\mathbb}

\newcommand{\R}{\mathbb{R}}
\newcommand{\mc}{\mathcal}

\newcommand{\D}{{\mc D}}

\newcommand{\Sc}{{\cal S}}

\newcommand{\F}{{\cal F}}
\newcommand{\G}{{\cal G}}
\newcommand{\Oc}{{\cal O}}

\newcommand{\Lc}{{\cal L}}

\newcommand{\U}{{\mathcal U}}
\newcommand{\V}{{\mathcal V}}
\newcommand{\1}{1 \!\! 1}
\newcommand{\ha}{\hat a}
\newcommand{\hb}{\hat b}

\newcommand{\Hil}{\mc H}

\newtheorem{thm}{Theorem}

\newtheorem{prop}[thm]{Proposition}
\newtheorem{defn}[thm]{Definition}

\newenvironment{proof}{\noindent {\bf Proof --}}{\hfill$\square$ \vspace{3mm}\endtrivlist}
\newcommand{\SC}{\mathcal{S}({\mb R})}

\newcommand{\Cv}{\frac{1}{\cos(2\nu)}}

\catcode `\@=11 \@addtoreset{equation}{section}
\catcode `\@=12

\begin{document}
\thispagestyle{empty}

\vspace*{2cm}

\begin{center}
{\Large \bf {Bi-squeezed states arising from pseudo-bosons}}   \vspace{2cm}\\

{\large F. Bagarello, F. Gargano, S. Spagnolo}\\
  Dipartimento di Energia, Ingegneria dell'Informazione e Modelli Matematici,\\
Scuola Politecnica, Universit\`a di Palermo,\\ I-90128  Palermo, Italy\\

\end{center}

\begin{abstract}
\noindent Extending our previous analysis on bi-coherent states,  we introduce here a new class of quantum mechanical vectors, the \emph{bi-squeezed states}, and we deduce  their main mathematical properties. We relate bi-squeezed states to the so-called regular and non regular pseudo-bosons. We  show that these two cases are  different, from a mathematical point of view. Some physical examples are considered.
\end{abstract}
\section{Introduction}
In last decades, the exigency to describe rigorously decaying quantum systems or systems going irreversibly towards a state of equilibrium has stimulated the research on quantum systems whose time evolution is ruled by non-hermitian Hamiltonians, see \cite{effH1,effH2,effH3} and references therein.
The introduction of this class of operators was useful to describe phenomenologically some kind of physical systems, neglecting the well known contradictions that the use of these Hamiltonians involve.
However, in the past twenty years, literature has increasingly focused its attention on  the possibility of having, in realistic situations and under specific conditions, Hamiltonians not necessarily hermitian but whose eigenvalues are real,  \cite{BB98}-\cite{BGV15}.  This is related to some symmetry conditions, physically motivated, that, like hermiticity, are again sufficient to guarantee reality of the spectrum, and possibly an unitary time evolution of the system, see \cite{bagdin1,bagdin2,bagdin3} and references therein.

This line of research has produced several results in quantum open systems, in quantum optics, in gain-loss systems and in other fields of quantum mechanics. Many application can be found in \cite{bagproc}.

From a mathematical point of view, losing Hermiticity of the Hamiltonian implies that the orthonormal (o.n) basis of its eigenvectors  must be replaced by two sets of biorthogonal states, no longer necessarily bases, \cite{baginbagbook}, but still complete.
In this context, in recent years, one of the authors, F.B., has considered in details some extended versions of the canonical (anti)-commutation relations, $[a,b]=\1$ (or $\{a,b\}=\1$), in which $a$ is not $b^{\dagger}$, and he has deduced several properties of the extended number-like operators, $N=ba$ and $N^{\dagger} (\neq N)$. In this approach, the definition of intertwining operators mapping the eigenstates of $N$ into those of $N^{\dagger}$ has been carried out in details, mainly considering the mathematical subtleties occurring when they happen to be unbounded. The related second-quantized framework produces the so-called pseudo-bosons and pseudo-fermions, \cite{baginbagbook}, or a nonlinear version of the first,  { \cite{BNLPbs,BZn1,BZn2}}. In connection with pseudo-bosons, the notion of bi-coherent states (BCS), originally introduced in \cite{tri} and then analysed, from a more mathematically oriented perspective, in  \cite{bagpb1}, has been considered in many of its aspects, see also \cite{BagQUO,bialo2017}. BCS can be considered as a non-hermitian generalization of coherent states, a class of quantum states playing a fundamental role both from a theoretical and an experimental point of view, \cite{gazeaubook, didier,aag}.
In this paper we  generalize a somehow related class of states, introducing the {\em bi-squeezed states} (BSS). They can be considered as a suitable extension of squeezed states, introduced originally in quantum mechanics in order to describe non-linear  processes such as optical parametric oscillations and four-wave mixing (see for example \cite{BarRad,ScullyZub,WallsMil} and references therein). As for the BCS, our main aim is to generalize  squeezed states to the context of non-hermitian quantum mechanics, similarly to what is done in \cite{tri2,maleki}, and to deduce  their main mathematical properties. 

{ The paper is organized as follows. In Section \ref{sec::prel}  we briefly review the relevant theory of BCS, and we discuss relations of these latter with a pseudo-bosonic structure. Section \ref{BSS} is dedicated to the definition of the deformed squeezing operators and of the BSS. In particular,  we consider both the case of regular and non regular BSS, for which there is no guarantee that squeezing operators are bounded. An application to the  Swanson model is described.
 Section \ref{sect::example} is devoted to an application of BSS in a quantum mechanical model ruled by a non hermitian Hamiltonian. Our conclusions are given in  Section \ref{sect::concl}.}

\section{Preliminaries}
\label{sec::prel}
To keep the paper self-contained, in this section we briefly review  the main features of pseudo-bosons and of BCS, putting in evidence those aspects which are particularly relevant for us.

\subsection{Some facts on $\mathcal D-$pbs}
\label{subsec:dpbs}
Let $\Hil$ be a given Hilbert space with scalar product $\left<.,.\right>$ and related norm $\|.\|$. Let  $\hat a$ and $\hat b$ be two operators
on $\Hil$, with domains $D(\hat a)$ and $D(\hat b)$ respectively, $\hat a^\dagger$ and $\hat b^\dagger$ their adjoint, and let $\D$ be a dense subspace of $\Hil$
such that $\hat a^\sharp\,\D\subseteq\D$ and $\hat b^\sharp\,\D\subseteq\D$, where $x^\sharp$ is $x$ or $x^\dagger$. Of course, $\D\subseteq D(\hat a^\sharp)\cap D(\hat b^\sharp)$.

\begin{defn}\label{def21}
The operators $(\hat a,\hat b)$ are $\D$-pseudo bosonic  if, for all $f\in\D$, we have
\be
[\hat a,\hat b]f=\hat a\,\hat b\,f-\hat b\,\hat a\,f=f,\quad \forall f\in \D.
\label{21}\en
\end{defn}
We  suppose that there exist two non-zero vectors $\varphi_{ 0}, \Psi_{ 0}\in\D$ such that
\bea\hat a\,\varphi_{ 0}=0,\quad \hat b^\dagger\Psi_{ 0}=0.\label{groundstate}\ena
It is clear that $\varphi_0\in D^\infty(\hat b):=\cap_{k\geq0}D(\hat b^k)$ and that $\Psi_0\in D^\infty(\hat a^\dagger)$, so
that we can define in $\D$ the  vectors
\be \varphi_n:=\frac{1}{\sqrt{n!}}\,{\hat b}^n\varphi_0,\qquad \Psi_n:=\frac{1}{\sqrt{n!}}\,({\hat a}^\dagger)^n\Psi_0, \label{22}\en
$n\geq0$, and the related sets   $\F_\Psi=\{\Psi_{ n}, \,n\geq0\}$,
$\F_\varphi=\{\varphi_{ n}, \,n\geq0\}$.
Since each $\varphi_n, \Psi_n\in \D ,\forall n\geq0$, they also belong to the domains of $a^\sharp$ and $b^\sharp$. Then we can deduce the following ladder relations:
 \begin{eqnarray}
\hat b\,\varphi_n&=&\sqrt{n+1}\varphi_{n+1},\quad  n\geq 0,\label{20} \\
\hat  a\varphi_n&=&\sqrt{n}\,\varphi_{n-1},\quad  n\geq 1,\label{201} \\
\hat a^\dagger\Psi_n&=&\sqrt{n+1}\Psi_{n+1},\quad n\geq 0,\label{202} \\
\hat b^\dagger\Psi_n&=&\sqrt{n}\,\Psi_{n-1},\quad  n\geq 1,
 \label{23}
  \end{eqnarray}
as well as the eigenvalue equations $\hat N\varphi_{ n}=n\varphi_{n}$ and $\hat N^\dagger\Psi_{ n}=n\Psi_{ n}$, where $N:=\hb\ha$ is the pseudo bosonic number operator. These imply, in particular, that, if we choose the normalization of $\varphi_0$ and $\Psi_0$ in such a way $\left<\varphi_0,\Psi_0\right>=1$, we get
\be \left<\varphi_n,\Psi_m\right>=\delta_{n,m}, \label{39}\en
 for all $n, m\geq0$, so that $\F_\varphi$ and $\F_\Psi$ are biorthogonal sets.

In concrete applications $\D-$pseudo-bosons (or, simply, pseudo-bosons) arise as \textit{deformations} of the standard bosonic operators, in the sense that there exists a non-unitary, but invertible, operator
$T$, not necessarily bounded, such that
\bea\varphi_n=Te_n,\quad \Psi_n=(T^{-1})^\dagger e_n, \quad n\geq0.\label{DPBS}\ena
Here $\F_e=\{e_n\in\D,n\geq0\}$ is an o.n. basis of $\Hil$.  If \eqref{DPBS} holds, the pseudo-bosonic operators are connected to standard bosonic annihilation and creation operators by the following similarity maps:
$$
\ha f=T\ha_0T^{-1}f,\quad \hb f=T\ha_0^\dag T^{-1}f,\quad \forall f\in \D,
$$
where $[\ha_0,\ha_0^\dagger]=\1$. In this case $\F_e$ is the usual o.n. basis connected with $\ha_0$ and its adjoint: $\ha_0 e_0=0$ and $e_n=\frac{1}{\sqrt{n!}}(\ha_0^\dagger)^ne_0$, $n\geq1$.
The mathematical treatment is simplified if $\D$ is left  invariant by  $T$ and $T^{-1}$, and by their adjoints.

If $T,T^{-1}$ are both bounded we get what has been called  \textit{regular} $\D-$pseudo-bosons, and
 $\F_\varphi,\F_\Psi$ are biorthogonal Riesz bases. If $T$ or $T^{-1}$ are not bounded, then pseudo-bosons are {\em non regular}, 
and  $\F_\varphi$ and $\F_\Psi$ are no longer biorthogonal Riesz bases. Sometimes, they are not even bases, but just complete sets. However, quite often
we can check  that there exists   a suitable dense subspace $\G$ of $\Hil$ such that, for all $f, g\in \G$, the following holds:
\be
\left<f,g\right>=\sum_{n\geq0}\left<f,\varphi_n\right>\left<\Psi_n,g\right>=\sum_{n\geq0}\left<f,\Psi_n\right>\left<\varphi_n,g\right>,
\label{25}
\en
which can be seen as a sort of resolution of the identity restricted to $\G$.
When (\ref{25}) is satisfied, $\F_{\varphi}$ and $\F_{\Psi}$ are called $\G-$quasi bases, \cite{baginbagbook}.

\subsection{Bi-coherent states}
\label{subsub:BCS}

From now on, also in view of our specific interest, we take $\Hil=\Lc^2(\Bbb R)$. Hence the vectors $\varphi_{ n}$, $\Psi_n$ and $e_n$  depend on a (spatial) variable $x$.
It is well known that the standard annihilation operator   $\hat a_0$ satisfying, as above, the canonical commutation relation $[\hat a_0,\hat a_0^\dagger]=\1$, admits a set of eigenstates $\Phi_z(x)$ labeled by a complex variable $z$. These eigenstates are called {\em coherent states} and they can be obtained through the action on the vacuum of $\hat a_0$, $e_0(x)$ ($\hat a_0\,e_0(x)=0$), of the unitary displacement operator
\be
W(z)=e^{z\hat a_0^\dagger-\overline{z}\,\hat a_0}=\sum_{k\geq0}\frac{1}{k!}\left(z\hat a_0^\dag -\overline{z}\hat a_0\right)^k,
\label{stUnitaryW}
\en
where the sums is convergent in $\D$, as follows:
\be
\Phi_z(x)=W(z)e_0(x)=e^{-|z|^2/2}\sum_{k=0}^\infty \frac{z^k}{\sqrt{k!}}\,e_k(x),\quad x\in \R.
\label{31}\en
It is known that 
\be
\hat a_0\,\Phi_z(x)=z\,\Phi_z(x),\qquad\mbox{and}\qquad \frac{1}{\pi}\int_{\Bbb C}d^2z|\Phi_z(x)\left>\right<\Phi_z(x)|=\1.
\label{add1c}\en
It is also well known that $\Phi_z(x)$ saturates the Heisenberg uncertainty relation.
In \cite{bagpb1, BagQUO,bialo2017} an  extension of coherent states was proposed and analyzed in a non-hermitian context. In particular the following states have been considered:
 \bea
 \varphi_z(x)&=&e^{-\frac{|z|^2}{2}}\sum_{n\geq0}\frac{z^n }{\sqrt{n!}}\varphi_{n}(x),\label{vphiBCS}\\ \Psi_z(x)&=&e^{-\frac{|z|^2}{2}}\sum_{n\geq0}\frac{z^n}{\sqrt{n!}}\Psi_{n}(x)\label{vpsiBCS}.
 \ena
 where $\varphi_n(x)$ and $\Psi_n(x)$ are the vectors of $\F_\varphi,\F_\Psi$, see (\ref{22}). It was shown that, under suitable conditions,
 \eqref{vphiBCS} and \eqref{vpsiBCS} are eigenstates of the pseudo-bosonic lowering operators $\ha$ and $\hb^\dagger$:
$$
\ha \varphi_z(x)=z\varphi_z(x),\qquad \hb^\dagger\Psi_z(x)=z\Psi_z(x),
$$
and satisfy the resolution of the identity
$$\frac{1}{\pi}\int_{\Bbb C}d^2z\left<f,\varphi_z\right>\left<\Phi_z,g\right>=\langle f,g \rangle,$$
for all $f,g\in\D$ if $\F_\varphi$ and $\F_\Psi$ are $\D$-quasi bases, or for all $f,g\in\Hil$ if $\F_\varphi$ and $\F_\Psi$ are (Riesz) bases.
It is then clear that $\varphi_z(x)$ and $\Psi_z(x)$ satisfy an extended version of the properties in \eqref{add1c} for ordinary coherent states. It is further possible to show that they saturate some deformed version of the  Heisenberg uncertainty relation. We will discuss this aspect later.

The states $\varphi_z(x)$ and $\Psi_z(x)$ can also be deduced via the action of two  displacement-like operators acting on the
vacua $\varphi_0(x)$ and $\Psi_0(x)$. To show this, it is convenient to work  under the assumption that a certain invertible operator $T$ exists, which is $\D$-invariant in the sense of \cite{BGST2017}. This means that $\D$ is invariant under the action of $T$, $T^\dagger$, and of their inverse. Then, as already observed above, it is possible to relate
 the pseudo-bosonic operators
$\hat a,\hat b$ to a pair of standard bosonic operator $\ha_0,\ha_0^{\dag}$ through
\bea
\ha f=T\ha_0T^{-1}f,\quad \hb f=T\ha_0^\dag T^{-1}f,\quad \forall f\in \D,\label{pseudo_standard}
\ena
where of course we are also assuming that $\ha_0$ and $\ha_0^{\dag}$  leave $\D$ stable as well\footnote{This is what happens, for instance,   if $\D$ is identified with $\Sc(\Bbb R)$, the set of test functions.}.
Similar equalities can be extended, if $T$ and $T^{-1}$ are both bounded, to two displacement-like operators $\U(z)$ and $\V(z)$ which we can define as follows:
\bea
\U(z)=TW(z)T^{-1},\qquad \V(z)=(T^{-1})^\dag W(z)T^{\dag}.\label{deformeddispla}
\ena
These operators are well defined and bounded for all $z\in\Bbb C$, since $T$, $T^{-1}$ and $W(z)$ are all bounded. Moreover, if $W(z)$ leaves $\D$ invariant, $\U(z)$ and $\V(z)$ do the same. In \cite{BagQUO}
 it has been proved that, for all $f\in\D$, the following series representation can be deduced for these operators:
\be
\U(z)f=\sum_{k=0}^\infty\frac{1}{k!}\left(z\hat b-\overline{z}\,\hat a\right)^kf, \qquad \V(z)f=\sum_{k=0}^\infty\frac{1}{k!}\left(z\hat a^\dagger-\overline{z}\,\hat b^\dagger\right)^kf,
\label{add1}\en
for all $f\in\D$, which shows that, despite the fact that $\hat a$ and $\hat b$ are unbounded, the series above converge strongly on $\D$ to $e^{z \hat b -\overline{z} \hat a}$ and to $e^{z {\hat a}^\dag -\overline{z}{\hat b}^\dag}$ respectively. In what follows, we simply write
 $$
 \U(z)=e^{z \hat b -\overline{z} \hat a}, \qquad \V(z)=e^{z {\hat a}^\dag -\overline{z}{\hat b}^\dag}.
 $$
Using now \eqref{DPBS} for $n=0$, formulas \eqref{deformeddispla} above, and the Baker-Campbell-Hausdorff formula for $W(z)$  we deduce the following alternative (and equivalent) expressions for our BCS:
 \hspace*{-4cm}\bea
\varphi_z(x)=\U(z)\varphi_0(x)=e^{z \hat b -\overline{z} \hat a}\left(Te_0(x)\right)=T\left(e^{-|z^2|/2}\sum_n\frac{z^n}{\sqrt{n!}}e_n(x)\right)=T\Phi_z(x),\\
\Psi_z(x)=\V(z)\,\Psi_0(x)=e^{z {\hat a}^\dag -\overline{z}{\hat b}^\dag}\left((T^{-1})^\dag e_0(x)\right)=(T^{-1})^\dag\left(e^{-|z^2|/2}\sum_n\frac{z^n}{\sqrt{n!}}e_n(x)\right)=\nonumber\\=(T^{-1})^\dag\Phi_z(x).
\label{BCS}
\ena


\subsubsection{Minimum uncertainty relation}
\label{mur}
Going back to the deformed version of the Heisenberg uncertainty relation cited above, we  introduce
the positive operator $\eta=(T^{-1})^\dag T^{-1}$, which is positive with positive inverse. We use $\eta$ to define the new scalar product $\left<\cdot,\cdot\right>_{\eta}=\left<\cdot,\eta\cdot\right>$.
 $\eta$ is usually called in the literature a metric operator. Now, given a (non necessarily hermitian) operator $\hat\Oc$, we define its (extended) uncertainty on the normalized vector $\chi\in\Hil$ {\em according to the new scalar product} as
 \bea
 (\Delta_{\eta}\hat \Oc)^2_\chi=
\left<\chi,\hat \Oc^2\,\chi \right>_{\eta}-\left<\chi,\hat \Oc\,\chi \right>_{\eta}^2. \label{variance}
 \ena
 Of course, if $\eta=\1$ and $\hat\Oc$ is hermitian, we recover the standard definition of uncertainty.
Then, if we introduce, following \eqref{pseudo_standard}, $\hat q=T\hat q_0 T^{-1}$ and $\hat p=T\hat p_0 T^{-1}$, where $\hat q_0=\frac{\ha_0+\ha_0^\dagger}{\sqrt{2}}$ and $\hat p_0=\frac{\ha_0-\ha_0^\dagger}{\sqrt{2}\,i}$ are the hermitian position and momentum operators, easy computations show that
$$
(\Delta_{\eta}\hat p)_{\varphi_z}=(\Delta\hat p_0)_{\Phi_z}, \qquad (\Delta_{\eta}\hat q)_{\varphi_z}=(\Delta\hat q_0)_{\Phi_z},
$$
 where, for instance, $(\Delta\hat p_0)_{\Phi_z}$ is the {\em standard} (i.e., with respect to the original scalar product) variance of $\hat p_0$ on the coherent state $\Phi_z(x)$ in (\ref{31}). Hence
 $$
 (\Delta_{\eta}\hat p)_{\varphi_z}(\Delta_{\eta}\hat q)_{\varphi_z}=(\Delta\hat p_0)_{\Phi_z}(\Delta\hat q_0)_{\Phi_z}=\frac{1}{2},
 $$
 due to the properties of  $\Phi_z(x)$. Then the deformed Heisenberg uncertainty relation for $\hat q$ and $\hat p$ (notice that $[\hat q,\hat p]f=f$, for all $f\in\D$) is saturated by $\varphi_z(x)$.  A similar conclusion can be deduced for the deformed variances of $\hat p^\dagger$ and $\hat q^\dagger$. In this case, however, rather than $\left<\cdot,\cdot\right>_{\eta}$ it is necessary to work with the $\left<\cdot,\cdot\right>_{\eta^{-1}}$ scalar product, which can be defined in complete analogy with $\left<\cdot,\cdot\right>_{\eta}$.
 
{ The importance of defining an appropriate scalar product is evident when dealing with some statistical properties associated to bi-coherent states.
 In fact, it is well known that  the coefficients $c_n=e^{-|z|^2/2}\frac{z^n}{\sqrt{n!}}$ in \eqref{31}
 define a Poissonian distribution, because $|c_n|^2=e^{-|z|^2}\frac{|z|^{2n}}{n!}$ and $\sum_{n\geq0}|\left<e_n(x),\Phi_z(x)\right>|^2=\sum_{n\geq0}|c_n|^2=1$.
 This means that $|c_n|^2$ is the measure of the probability of detecting $n$ quanta per time interval if $|z|^2$ is the average number of quanta.
 Moreover it is known that the uncertainty of the number operator $\hat N_0=\hat a_0^{\dag}\ha_0$ over a coherent state is given by the relation  $(\Delta\hat N_0)^2_{\Phi_z}=|z|^2$.
 Of course, these properties are  direct consequence of the fact that the states $e_n(x)$ are orthonormal, a condition which is not satisfied by the states $\varphi_n(x)$ and $\Psi_n(x)$: 
 this simply implies that, in general, $\sum_{n\geq0}|\left<e_n(x),\varphi_z(x)\right>|^2\neq1,\quad \sum_{n\geq0}|\left<e_n(x),\Psi_z(x)\right>|^2\neq1$, and hence no Poissonian distribution can be retrieved.
 However, considering the modified scalar product $\left<\cdot,\cdot\right>_{\eta}$ and the sets of bi-orthonormal states used to build the bi-coherent states, we  find 
 $$\sum_{n\geq0}|\left<\varphi_n(x),\varphi_z(x)\right>_{\eta}|^2=\sum_{n\geq0}|\left<T^{-1}\varphi_n(x),T^{-1}\varphi_z(x)\right>|^2=\sum_{n\geq0}|\left<e_n(x),\Phi_z(x)\right>|^2=\sum_{n\geq0}|c_n|^2=1.
 $$
 Analogously   $$\sum_{n\geq0}|\left<\Psi_n(x),\Psi_z(x)\right>_{\eta^{-1}}|^2=1.$$

 Moreover, using \eqref{variance},
 $$
(\Delta_{\eta}\hat N)^2_{\varphi_z(x)}=(\Delta_{\eta^{-1}}\hat N^{\dag})^2_{\Psi_z(x)}=(\Delta\hat N_0)^2_{\Phi_z(x)}=|z|^2,
 $$
 where $\hat N=T^{-1} \hat N_0 T$ is the pseudo bosonic number operator.  
 Hence, we recover here similar statistical interpretation as for the  coherent states. However, the price to pay is that we need to deform the scalar product accordingly to the state we are considering. Which is not necessarily the best one can expect.
 
%
 }

\section{Bi-squeezed states}
\label{BSS}

In this section, after a short review of some well known properties of  squeezed states, we analyze  the particular case in which regular BSS arise from the application of a bounded operator $T$, with bounded inverse, on a standard squeezed state.

In analogy with what we have done in Section \ref{subsub:BCS}, we will work  under the assumption that $T$ is $\D$-invariant, and
we further refine our assumptions by requiring that $S(z)$, see  equation \eqref{sqeeoperator} below,  leaves  $\D$ invariant, too. In Section \ref{sec::ubounded} we will also briefly discuss what happens when $T$ or $T^{-1}$ are unbounded.

\subsection{Standard squeezed states}\label{sss}

Squeezed states are a class of minimum-uncertainty states that are strongly connected to coherent states. The main difference between  coherent and squeezed states is that for the latter the noise in the quadratures can be different  while for the former is equal, see \cite{WallsMil}. Squeezed states play a very important role, for instance, in quantum optics (in non-linear phenomena as optical parametric oscillation and four-wave mixing, \cite{WallsMil}), and in quantum electrodynamics (for example in dynamical Casimir effect, \cite{Nori2011}).


Squeezed states are defined by introducing first the standard unitary squeezing operator
\bea S(z)=e^{\frac{z}{2}(\hat a_0^{\dag})^2-\frac{\overline{z}}{2}(\hat a_0)^2},\label{sqeeoperator}\ena
 $z\in\Bbb C$, and then the normalized
squeezed state by its action on $e_0(x)$,
\bea
\psi^0_{z}(x)=S(z)\,e_{0}(x).\label{standard_squizzo}
\ena
Sometimes it is convenient to rewrite $S(z)$ in a factorized form as follows:
\be S(z)=e^{\lambda_b(z) (\hat a_0^\dag)^2}e^{\lambda(z) (\ha_0\ha_0^\dag+\ha_0^\dag \ha_0)}e^{\lambda_a(z) \ha_0^2},\label{add1}\en
where $z=re^{i\theta}$, $\lambda(z)=-\frac{1}{2}\log(\cosh r) , \lambda_a(z)=\frac{1}{2}e^{-i\theta}\tanh r$, and $\lambda_b(z)=-\frac{1}{2}e^{i\theta}\tanh r=-\overline{\lambda_a(z)}$. The  factorization of $S(z)$ allows us to express the squeezed state as
\be
 \psi^0_z(x)=S(z)e_{0}(x)=e^{\lambda(z)} \sum_{k=0}^\infty\left(\lambda_b(z)\right)^k\frac{\sqrt{(2k!)}}{k!}e_{2k}(x),
 \label{squizzi0}
\en
which is uniformly convergent, $\forall z\in\Bbb C$.
Moreover, the coherent squeezed states, defined as
\be\psi^\alpha_{z}(x)=W(\alpha)S(z)e_{0}(x),\label{css}\en
can also be introduced. This is  the result of the successive  applications of the displacement and of the squeezing operators on the vacuum $e_0(x)$. Well known features of coherent squeezed states, \cite{BarRad}, are the following:
\bea
 \left[\cosh r \left(\ha_0-\alpha\right)+\exp{(i\theta)}\sinh r \left(\ha_0^{\dag}-\overline{\alpha}\right)\right]\psi^\alpha_{z}(x)&=&0,\label{squeezeigen}\\
 \left(\ha_0+z \ha_0^\dag\right)\psi^\alpha_{z}(x)&=&\alpha\psi^\alpha_{z}(x),\label{squeezeigen2}\\
 \langle \psi^\alpha_{z},\psi^\alpha_{z}\rangle&=&1 \label{csorto},\\
\frac{1}{\pi}\int_{\Bbb C} d\alpha \left<f,\psi_{z}^{\alpha}\right>\left< \psi_{z}^{\alpha},g\right >&=&\left< f,g\right >,\label{sidentiti}
\ena
for all $f,g\in\Hil$. Then, the vectors $\psi^\alpha_{z}(x)$ are normalized and resolve the identity. Notice, however, that the set $\{\psi^0_{z}(x)\}$ does not! It is only the presence of $\alpha$, and of the related displacement operator, which guarantees the validity of equation (\ref{sidentiti}). 

From 
\eqref{squeezeigen} we observe that $\psi^\alpha_{z}(x)$ is the vacuum of the operator $$A=\left[\cosh r \left(\ha_0-\alpha\right)+\exp{(i\theta)}\sinh r \left(\ha_0^{\dag}-\overline{\alpha}\right)\right],$$ which incidentally satisfies the commutation rule  $[A,A^\dagger]=\left((\cosh r)^2-(\sinh r)^2\right)\1=\1$.

\subsection{Doubling the squeezing operator}
\label{sec::Squeezing operators}

Extending now what we have shown in Section \ref{subsub:BCS},  
we   prove the existence of a pair of deformed squeezing operators and we discuss their relation with the standard squeezing operator
 $S(z)$ through a similarity operation which involves the same operator $T$ appearing, for instance, in (\ref{DPBS}) and in (\ref{deformeddispla}).

We first define, $\forall z\in \mathbb{C}$, the operators $\Sc(z)$ and $\mathcal{T} (z)$ as follows:
 \be \Sc(z)f=T S(z)T^{-1}f,\quad \mathcal{T} (z)f=(T^{-1})^\dag S(z)T^{\dag}f,\label{squeezeoperators}\en
 $\forall f\in\D$.
 Of course the above definitions are well posed (in fact, both $\Sc(z)$ and $\mathcal{T} (z)$ are bounded), and produce results in $\D$, since $T$ is $\D-$stable, \cite{BGST2017}, and $S(z)$ leaves $\D$ invariant. The definitions of $\Sc(z)$ and $\mathcal{T} (z)$ are suggested by the analogous definitions adopted for the displacement operators in \eqref{deformeddispla} and, as in \eqref{deformeddispla},  formulas \eqref{squeezeoperators} can be extended to all of $\Hil$. In this way we can get  the following intertwining relation between $\mathcal{T}(z)$ and $\mathcal{S}(z)$:
 $$
 TT^\dag\mathcal{T}(z)=\Sc(z)TT^\dag.
 $$
 Intertwining relations are quite relevant in connection with quantum solvable models, \cite{intop}. However, this is not main interest here and we will not consider further this aspect.

 It is  possible to describe the actions of $\Sc(z), \mathcal{T} (z)$ in terms of convergent series. 

\begin{prop}\label{propBSQOP1}
The following equalities holds:
\be \Sc(z)f= \sum_{k\geq0}\frac{1}{k!}\left(\frac{z}{2}\hb^2-\frac{\overline{z}}{2}\ha^2\right)^{k}f,\quad \mathcal{T} (z)f= \sum_{k\geq0}\frac{1}{k!}\left(\frac{z}{2}(\ha^\dagger)^2-\frac{\overline{z}}{2}(\hb^\dagger)^2\right)^{k}f,\label{propBSQOP2}\en
for all {$f\in\mathcal{D}$}.
\end{prop}
\begin{proof}
We first prove that, for all $f\in \D$, and for all $k\in\mathbb{N}$,
\bea T\left(\frac{z}{2}(\hat a_0^{\dag})^2-\frac{\overline{z}}{2}\hat a_0^2\right)^kT^{-1}f=\left(\frac{z}{2}\hat b^2-\frac{\overline{z}}{2}\hat a^2\right)^kf.\label{st1}
\ena
For $k=0$ the equality is evident.
For $k=1$ the proof follows from \eqref{pseudo_standard} and from the stability of $\D$:
$$
T\left(\frac{z}{2}(\ha_0^{\dag})^2-\frac{\overline{z}}{2}\ha_0^2\right)T^{-1}f=
\left(\frac{z}{2}T\ha_0^{\dag}T^{-1}T\ha_0^{\dag}T^{-1}-\frac{\overline z}{2}T\ha_0T^{-1}T\ha_0T^{-1}\right)f=
\left(\frac{z}{2}\hb^2-\frac{\overline{z}}{2}\ha^2\right)f.
$$

Now assuming that \eqref{st1}
 holds for $k$, and recalling that all the operators are $\D$-stable, we get
\beano 
&&T\left(\frac{z}{2}(\ha_0^{\dag})^2-\frac{\overline{z}}{2}\ha_0^2\right)^{k+1}T^{-1}f =T\left(\frac{z}{2}(\ha_0^{\dag})^2-\frac{\overline{z}}{2}\ha_0^2\right)^kT^{-1}T
\left(\frac{z}{2}(\ha_0^{\dag})^2-\frac{\overline{z}}{2}\ha_0^2\right)T^{-1}f=\\
&&=T\left(\frac{z}{2}(\ha_0^{\dag})^2-\frac{\overline{z}}{2}\ha_0^2\right)^kT^{-1}\left(\frac{z}{2}\hb^2-\frac{\overline{z}}{2}\ha^2\right)f=\\
&&=\left(\frac{z}{2}\hb^2-\frac{\overline{z}}{2}\ha^2\right)^{k}\left(\frac{z}{2}\hb^2-\frac{\overline{z}}{2}\ha^2\right)f=\\
&&=\left(\frac{z}{2}\hb^2-\frac{\overline{z}}{2}\ha^2\right)^{k+1}f.
\enano
This is a consequence of our induction hypothesis, and of the fact that $\left(\frac{z}{2}\hb^2-\frac{\overline{z}}{2}\ha^2\right)f\in\D$.
Since $S(z)$ satisfies the following expansion:
$$
S(z)\tilde f=\sum_{k=0}^\infty\frac{1}{k!}\left(\frac{z}{2}(\ha_0^{\dag})^2-\frac{\overline{z}}{2}\ha_0^2\right)^{k}\tilde f
$$
for all $\tilde f\in\D$, the continuity of $T$ implies that
$$
TS(z)\tilde f=\sum_{k=0}^\infty\frac{1}{k!}T\left(\frac{z}{2}(\ha_0^{\dag})^2-\frac{\overline{z}}{2}\ha_0^2\right)^{k}\tilde f
$$
for all such $\tilde f$. Moreover, $\tilde f$ can be written as $T^{-1}T\tilde f=T^{-1}f$, where $f=T\tilde f\in\D$. Then we deduce that
\beano \Sc(z)f=TS(z)T^{-1}f=
\sum_{k\geq0}\frac{1}{k!}\left(\frac{z}{2}\hb^2-\frac{\overline{z}}{2}\ha^2\right)^{k} f,
\enano
for all $f\in\D$, as we had to prove. The proof for $\mathcal{T}(z)$ is   similar.
\end{proof}
Despite of the unboundedness of $\hat a$ and $\hat b$, the series in \eqref{propBSQOP2} converge strongly on $\D$, and as we did for the deformed displacements operators, from now on we simply write
 \bea
 \mathcal{\Sc} (z)&=&e^{ \frac{1}{2}z\hb^2-\frac{1}{2}\bar{z}\ha^2},\label{squeez_S}\\
 \mathcal{T} (z)&=&e^{\frac{1}{2}z(\ha^\dag)^2-\frac{1}{2}\bar{z}(\hb^\dag)^2}\label{squeez_T}.
 \ena
These operators satisfy the following relations:
	\be \Sc^{-1}(z)=\Sc(-z)=\mathcal{T} ^\dag(z),\quad \mathcal{T}^{-1}(z)=\mathcal{T}(-z)=\Sc^\dag(z). \label{squeeInv}\en

Incidentally we observe that an alternative (formal) representation of the above operators can be deduced using the  Baker-Campbell-Hausdorff formula. We get
\bea
\mathcal{\Sc} (z)&=&e^{\lambda_b(z) \hb^2}e^{\lambda(z) (\ha \hb+\hb \ha)}e^{\lambda_a(z) \ha^2}=e^{\lambda_a(z) \ha^2}e^{-\lambda(z) (\ha \hb+\hb \ha)}e^{\lambda_b(z) \hb^2},\label{bch1}\\
\mathcal{T} (z)&=&e^{\lambda_b(z) (\ha^\dag)^2}e^{\lambda(z) (\hb^\dag \ha^\dag+\ha{^\dag}\hb^\dag)}e^{\lambda_a(z) (\hb^\dag)^2}=e^{\lambda_a(z) (\hb^\dag)^2}e^{-\lambda(z) (\hb^\dag \ha^\dag+\ha{^\dag}\hb^\dag)}e^{\lambda_b(z) (\ha^\dag)^2},\nonumber\\
\label{bch2}
\ena
which are the deformed versions of equation \eqref{add1} for $S(z)$.
\vspace{1mm}

{\bf Remark:--} The reason why we call these formulas {\em formal} is because, while $\Sc(z)$ and $\mathcal{T}(z)$ are bounded, the single terms in (\ref{bch1}) and (\ref{bch2}) are not. Hence, for instance, there is no guarantee a priori that $e^{\lambda_b(z) \hb^2}$ is densely defined, or leaves $\D$ invariant, or that, at least, maps $\D$ into the domain of $e^{\lambda(z) (\ha \hb+\hb \ha)}$.

\subsection{Regular bi-squeezed states}
\label{sec::Bi-squeezed states}

We are now ready to define a pair of  states,  $\tau_z(x)$ and $\kappa_z(x)$, which, as we will shown later, can be considered a natural extension of the standard squeezed state in
 \eqref{standard_squizzo}.

\begin{defn}\label{RBSS}
A pair of states $\left(\tau_z(x),\kappa_z(x)\right)$, $x\in \R$, $z\in\mathbb{C},$ are called $\mathbb{C}$-regular BSS ($\mathbb{C}$-RBSS) if there exist a squeezed state $\psi^0_z(x)\in \D$, and a bounded operator $T$ with bounded inverse $T^{-1}$, $\D$-stable,  such that
\bea
\tau_z(x)=T\psi^0_z(x),\quad \kappa_z(x)=(T^{-1})^\dag\psi^0_z(x).\label{rbss_def}
\ena
\end{defn}
It  is clear, first of all,  that $\tau_z(x),\kappa_z(x)\in\D$. Moreover,  for all
$z\in\mathbb{C}$, 
$\lVert\tau_z \rVert=\lVert T\psi^0_z \rVert\leq \lVert T\rVert\lVert\psi^0_z \rVert\leq \lVert T\rVert $ and 
$\lVert\kappa_z \rVert\leq \lVert T^{-1}\rVert $. We shall see in Section \ref{sec::ubounded}
that a similar definition is not the most convenient  when $T$ or $T^{-1}$ are unbounded.

It is interesting to observe  how these states are related to the operators $\Sc(z)$ and $\mathcal{T}(z)$. This is what the next proposition is about.

\begin{prop}\label{squuezeprop}
Let  $\left(\tau_z(x),\kappa_z(x)\right)$, be a pair of $\mathbb{C}$-RBSS, and
$\varphi_0(x),\Psi_0(x)$ the two vacua in  \eqref{groundstate}.
Then,
\bea
\tau_z(x)=\Sc(z)\varphi_0(x),\quad \kappa_z(x)=\mathcal{T}(z)\Psi_0(x).\label{bcsT}
\ena
Moreover they satisfy the bi-normalization condition
\bea
\langle \tau_z(x),\kappa_z(x)\rangle=1.\label{biosqueeze}
\ena
\end{prop}
\begin{proof}
Using the boundedness of $T$, $T^{-1}$ and $S(z)$, we have
$$\tau_z(x)=T\psi_z^0(x)=T\left(S(z)e_0(x)\right)=\left(TS(z)T^{-1}\right)\left(Te_0(x)\right)=\mathcal{\Sc}(z)\varphi_0(x).$$
Similarly we prove that $\kappa_z(x)=\mathcal{T}(z)\Psi_0(x)$. Now, since $\Sc^\dagger(z)=\mathcal{T}^{-1}(z)$,
$$
\langle \tau_z,\kappa_z\rangle=\langle\Sc(z)\varphi_0,\mathcal{T}(z)\Psi_0\rangle=\langle \varphi_0,\mathcal{T}^{-1}(z)\mathcal{T}(z)\Psi_0\rangle=\langle \varphi_0,\Psi_0\rangle=1,
$$
due to (\ref{39}).

\end{proof}

The above proposition states that $\left(\tau_z(x),\kappa_z(x)\right)$, originally introduced as in (\ref{rbss_def}), can also be obtained applying the deformed squeezing
operators $\mathcal{S}(z)$ and $\mathcal{T}(z)$ over, respectively, $\varphi_0(x)$ and $\Psi_0(x)$. 
Moreover, using the continuity of $T,(T^{-1})^\dagger$, and the expansion \eqref{squizzi0},  it is straightforward to express $\tau_z(x)$ and $\kappa_z(x)$ still in a different way:
\begin{eqnarray}
&&\tau_z(x)=T\psi^0_z(x)=T\left[e^{\lambda(z)} \sum_{k=0}^\infty\left(\lambda_b(z)\right)^k\frac{\sqrt{(2k!)}}{k!}e_{2k}(x)\right]=
e^{\lambda(z)} \sum_{k=0}^\infty\left(\lambda_b(z)\right)^k\frac{\sqrt{(2k!)}}{k!}Te_{2k}(x)= \nonumber \\
&& =e^{\lambda(z)}\sum_{k=0}^\infty\lambda_b(z)^k\frac{\sqrt{(2k!)}}{k!}\varphi_{2k}(x),\label{series1}\\
&&\kappa_z(x)=(T^{-1})^\dagger \psi^0_z(x)=(T^{-1})^\dagger \left[e^{\lambda(z)} \sum_{k=0}^\infty\left(\lambda_b(z)\right)^k\frac{\sqrt{(2k!)}}{k!}e_{2k}(x)\right]=\nonumber\\
&& =e^{\lambda(z)} \sum_{k=0}^\infty\left(\lambda_b(z)\right)^k\frac{\sqrt{(2k!)}}{k!}(T^{-1})^\dagger e_{2k}(x)= e^{\lambda(z)}\sum_{k=0}^\infty\left(\lambda_a(z)\right)^k\frac{\sqrt{(2k!)}}{k!}\Psi_{2k}(x).\label{series2}
\end{eqnarray}


These expansions will appear to be particularly relevant in Section \ref{sec::ubounded}, in connection with non regular pseudo-bosons, i.e. with the case in which $T$ or $T^{-1}$ are unbounded and, therefore, not continuous.

\vspace{2mm}

{\bf Remark:--} It is possible to show that these vectors are stable under time evolution, at least if we assume a pseudo-bosonic number operator for the Hamiltonian of the system. Let, in fact, $H=ba$. Then, if we can bring the  operator $e^{-iHt}$ inside the infinite sum\footnote{This is not granted, since this operator is not unitary.}, we get 
$$
\tau_z(x,t)=e^{-iHt}\tau_z(x)=e^{\lambda(z)}\sum_{k=0}^\infty\lambda_b(z)^k\frac{\sqrt{(2k!)}}{k!}e^{-2ikt}\varphi_{2k}(x).
$$
Now, recalling that $\lambda_b(z)=-\frac{1}{2}e^{i\theta}\tanh r$ and that $\lambda(z)$ does not depend on $\theta$, we conclude that $\tau_z(x,t)$ coincides with $\tau_z(x)$, but with $\theta$ replaced by $\theta-2t$. This implies that the time evolution of a squeezed state is still a squeezed state. The same conclusion can be deduced, not surprisingly, also for the time evolution of $\kappa_z(x,t)$.
Of course  with similar arguments it can be shown that also the bi-coherent states
are stable under time evolution, extending the work done in \cite{maa} in the framework of the pseudo-fermionic operators.

\vspace{3mm}

{\bf Example: the deformed harmonic oscillator}

\vspace{1mm}

We want to show now how BSS look like for a very simple system.
Consider the harmonic oscillator and its Hamiltonian $H_0=\ha_0^\dagger\ha+\frac{1}{2}\1$. Its ground state is $e_0(x)=\frac{1}{\pi^{1/4}}\textrm{exp}(-\frac{1}{2}x^2)$.
We introduce two function $u,v\in\Sc(\mathbb{R})$, satisfying  $\langle u,v\rangle=1$,
and two complex scalar $\alpha,\beta$ satisfying  $\alpha+\beta+\alpha\beta=0$. Let $P_{u,v}$ be the operator defined  as $P_{u,v}f=\langle u,f\rangle v$, for all $f\in\Lc^2(\Bbb R)$, and let $T$ be
the operator $T=\1+\alpha P_{u,v}$. Then $T$ is bounded with bounded inverse $T^{-1}=\1+\beta P_{u,v}$.
The operator $T$ was already considered in \cite{BGST2017}, where it was proved to be  $\Sc(\R)$-stable and to define the following biorthogonal Riesz bases
\beano \F_\varphi&=&\{\varphi_n(x)=T e_n(x)=e_n(x)+\alpha\left<u,e_n\right>v(x)\}, \\
\F_\Psi&=&\{\Psi_n(x)=(T^{-1})^\dag e_n(x)=e_n(x)+\overline{\beta}\left<v,e_n\right>u(x)\}.\enano
Here  $\F_e=\{e_n(x)\in\SC\}$ is the o.n. basis of $\Lc^2(\Bbb R)$ of eigenstates of $H_0$.
The functions $\varphi_n(x)$ and $\Psi_n(x)$ are in $\SC$ as well. The $\mathbb{C}$-RBSS turn out to be 
\beano
\tau_z(x)=e^{\lambda(z)}\sum_{k\geq0}\frac{\lambda_b(z)^k}{k!}\sqrt{(2k)!}\left(e_{2k}(x)+\alpha\left<u,e_{2k}\right>v(x)\right)=\psi^0_{z}(x)+\alpha\left<u,\psi^0_{z}\right>\,v(x),\\ \kappa_z(x)=e^{\lambda(z)}\sum_{k\geq0}\frac{\lambda_b(z)^k}{k!}\sqrt{(2k)!}\left(e_{2k}(x)+\overline{\beta}\left<v,e_{2k}\right>u(x)\right)=\psi^0_{z}(x)+\overline{\beta}\left<v,\psi^0_{z}\right>\,u(x),
\enano
where $\psi^0_{z}(x)$ is the standard squeezed state  in (\ref{squizzi0}). Hence, for the deformed harmonic oscillator with Hamiltonian $H=\hat b\hat a+\frac{1}{2}\,\1$, see (\ref{pseudo_standard}), the BSS are simply two suitable linear combinations of $\psi^0_{z}(x)$ with $v(x)$ and with $u(x)$ respectively, with coefficients which are related to $\psi^0_{z}(x)$ itself.

\vspace{3mm}

Formula (\ref{sidentiti}) shows that $\psi^0_{z}(x)$   alone is not enough to produce a resolution of the identity. We also need to use the displacement operator. This is the reason why we introduce now the following definition:

\begin{defn}\label{RCBSS}
Let $\alpha\in\mathbb{C}$, $x\in \R$.
A pair of $\mathbb{C}$-RBSS $\left(\tau_{z}^{\alpha}(x),\kappa_{z}^{\alpha}\right(x))$,  are called $\mathbb{C}$-regular coherent BSS ($\mathbb{C}$-RCBSS), if there exist a coherent squeezed state $\psi^\alpha_z(x)\in\D$, \eqref{css}, and a bounded $\D$-stable operator $T$, with bounded inverse $T^{-1}$,  such that
\bea
\tau_{z}^{\alpha}(x)=T\psi^\alpha_z(x),\quad \kappa_{z}^{\alpha}(x)=(T^{-1})^\dag\psi^\alpha_z(x).\label{rcbss_def}
\ena

\end{defn}

It is clear that $\tau_{z}^{0}(x)=\tau_{z}(x)$ and $\kappa_{z}^{0}(x)=\kappa_{z}(x)$, see \eqref{rbss_def}, and that  $\tau_{z}^{\alpha}(x)$ and $\kappa_{z}^{\alpha}(x)$ are in $\D$.
It is also  easy to extend \eqref{bcsT}.  In fact, $\forall \alpha\in \mathbb{C},$ we deduce that
\bea
\tau_{z}^{\alpha}(x)=TW(\alpha)S(z)e_{0}=\left(TW(\alpha)T\right)\left(T^{-1}S(z)T\right)T^{-1}e_{0}(x)=\mathcal{U(\alpha)}\Sc(z)\varphi_0(x),
\ena
and analogously
\bea
\quad \kappa_{z}^{\alpha}(x)=\mathcal{V(\alpha)}\mathcal{T}(z)\Psi_0(x).\label{bcsdT}
\ena
The following proposition can now be proved:
\begin{prop}\label{propBSQOP22}
Let  $\left(\tau_{z}^{\alpha}(x),\kappa_{z}^{\alpha}(x)\right)$ be a pair of $\mathbb{C}$-RCBSS.
The following equalities hold, $\forall \alpha\in\mathbb{C}$, $\forall z\in\mathbb{C}$, and $\forall f,g\in\Hil$ :
\bea
&& \left[\cosh r \left(\ha-\alpha \right)+\exp{(i\theta)}\sinh r\left(\hb-\overline{\alpha} \right) \right]\tau^\alpha_{z}(x)=0\label{247}\\
 && \left[\cosh r \left(\hb^\dagger-\alpha \right)+\exp{(i\theta)}\sinh r\left(\ha^\dagger-\overline{\alpha} \right) \right]\kappa^\alpha_{z}(x)=0,\label{248}\\
  && \left(\ha+z \hb\right)\tau^\alpha_{z}(x)=\alpha\tau^\alpha_{z}(x),\quad 
  \left(\hb^\dag+z \ha^\dag\right)\kappa^\alpha_{z}(x)=\alpha\kappa^\alpha_{z}(x)\label{squeezeigen2alpha}\\
  && \langle \tau_{z,\alpha}(x),\kappa_{z,\alpha}(x)\rangle=1, \label{249}\\
  && \frac{1}{\pi}\int_{\mathbb{C}}d\alpha\langle f, \tau_{z}^{\alpha}(x)\rangle\langle\kappa_{z}^{\alpha}(x),g\rangle=\langle f,g \rangle,\quad   \label{243}
\ena
\end{prop}

\begin{proof}\label{proofBSQOP23}
To prove \eqref{247} we use \eqref{pseudo_standard} and the fact that $\tau^\alpha_{z}(x)\in\D$,
\begin{eqnarray}
 &&\left[\cosh r \left(\ha-\alpha \right)+\exp{(i\theta)}\sinh r\left(\hb-\overline{\alpha} \right) \right]\tau^\alpha_{z}(x)=\nonumber\\
  && T\left[\cosh r \left(\ha_0 -\alpha \right)+\exp{(i\theta)}\sinh r\left(\ha_0^\dagger -\overline{\alpha} \right) \right]T^{-1}\tau^\alpha_{z}(x)=\nonumber\\
 && T \left[\cosh r \left(\ha_0 -\alpha \right)+\exp{(i\theta)}\sinh r\left(\ha_0^\dagger-\overline{\alpha} \right) \right] \psi^\alpha_{z}(x)=0,\nonumber
\end{eqnarray}
by equation (\ref{squeezeigen}).
In the same way, but using the deformation given by $T^\dagger,(T^{-1})^\dagger$ we can  prove
 \eqref{248}

Formulas in \eqref{squeezeigen2alpha} can be proved as follows:
\beano
\left(\ha+\alpha \hb\right)\tau^\alpha_{z}(x)=\left(T\ha_0T^{-1}+\alpha T\ha_0^\dag T^{-1}\right)\tau^\alpha_{z}(x)=\\
T\left(\ha_0+\alpha \ha_0^\dag \right)T^{-1}\tau^\alpha_{z}(x)=T\left(\ha_0+\alpha \ha_0^\dag \right)\psi^\alpha_{z}(x)=\\
T\left(z\psi^\alpha_{z}(x)\right)=z\tau^\alpha_{z}(x).
\enano
The proof for $\kappa_z^\alpha(x)$ is analogous.

The bi-normalization condition \eqref{249} easily follows from \eqref{csorto}:
\be
\   \langle \tau_{z}^{\alpha}(x),\kappa_{z}^{\alpha}(x)\rangle=\langle T\psi^\alpha_{z}(x),(T^{-1})^{\dag}\psi^\alpha_{z}(x)\rangle=\langle T^{-1}T\psi^\alpha_{z}(x)\psi^\alpha_{z}(x)\rangle=1.
\en
To prove \eqref{243} we use the resolution of the identity \eqref{sidentiti}, valid $\forall z\in\mathbb{C}$, and for all $f,g\in\Hil$:
\begin{eqnarray}\label{CloRel}
&&\langle f,g\rangle=\langle f,T(T^{-1}) g\rangle=\langle T^{\dag}f,(T^{-1}) g\rangle=\nonumber \\
&&=\frac{1}{\pi}\int_{\mathbb{C}}d\alpha\langle T^{\dag}f, \psi^\alpha_{z}(x)\rangle\langle\psi^\alpha_{z}(x), (T^{-1}) g\rangle=\frac{1}{\pi}\int_{\mathbb{C}}d\alpha\langle f, T\psi^\alpha_{z}(x)\rangle\langle(T^{-1})^{\dag}\psi^\alpha_{z}(x),g\rangle=\nonumber \\&&=\frac{1}{\pi}\int_{\mathbb{C}}d\alpha\langle f, \tau_{z}^{\alpha}(x)\rangle\langle\kappa_{z}^{\alpha}(x),g\rangle
\end{eqnarray}

\end{proof}

Summarizing, we have shown that our vectors have  properties which are very similar to those of the ordinary squeezed states. The differences arise mainly as a consequence of the different contexts in which these states are considered (ordinary or PT quantum mechanics). 

\subsection{Some results on non regular bi-squeezed states}
\label{sec::ubounded}
In the previous sections we have eavily used the hypothesis that $T$ and $T^{-1}$ are bounded. This has produced, for instance, a series expression for the squeezing operators and for the related regular BSS, see Proposition \ref{propBSQOP1} and equations (\ref{series1}) and (\ref{series2}).
We can in general extend the definitions of the operators and of the squeezed states also to an unbounded  $T$ (or $T^{-1}$).
Of course, in this case there is no guarantee that the squeezing operators $\Sc(z)$ and $\mathcal T(z)$   are bounded, and in fact, in general, they are not. 

For this reason, rather than trying to apply formula \eqref{bcsT},   it is more natural to define bi-squeezed states   as in \eqref{series1}-\eqref{series2} through the series expansions containing the vectors of $\F_{\varphi},\F_{\Psi}$:
\begin{eqnarray}
\tau_z(x)=e^{\lambda(z)}\sum_{k\geq0}\frac{\lambda_b(z)^k}{k!}\sqrt{(2k)!}\varphi_{2k}(x),\quad \kappa_z(x)=e^{\lambda(z)}\sum_{k\geq0}\frac{\lambda_a(z)^k}{k!}\sqrt{(2k)!}\Psi_{2k}(x),\label{bisquizzi}
\end{eqnarray}
and check for convergence conditions for these series. This is exactly what we have done, for instance, for bi-coherent states in \cite{bialo2017}.
It is not a big surprise that convergence of the above series is not guaranteed  in all of $\mathbb{C}$, in this case.  In fact,  we can prove the following result, giving sufficient conditions for the series above to converge.

\begin{thm}\label{theo2}
 Consider a sequence of complex numbers $\alpha_n\neq0,\forall n\geq0,$ such that $\lim_{n\rightarrow\infty}\left|\frac{\alpha_{n+1}}{\alpha_n}\right|=\overline\alpha$.
Assume that four strictly positive constants $A_\varphi$, $A_\Psi$, $r_\varphi$ and $r_\Psi$ exist, together with two strictly positive sequences $M_n(\varphi)$ and $M_n(\Psi)$ for which
\be
\lim_{n\rightarrow\infty}\frac{M_n(\varphi)}{M_{n+2}(\varphi)}=M(\varphi), \qquad \lim_{n\rightarrow\infty}\frac{M_n(\Psi)}{M_{n+2}(\Psi)}=M(\Psi),
\label{30}\en
where $M(\varphi)$ and $M(\Psi)$ could be infinity, such that, for all $n\geq0$,
\be
\|\varphi_n\|\leq A_\varphi\,r_\varphi^n M_n(\varphi), \qquad \|\Psi_n\|\leq A_\Psi\,r_\Psi^n M_n(\Psi).
\label{31b}\en
Then, the following series:
\be
\sum_{n=0}^\infty\frac{\lambda_b(z)^n}{\alpha_n}\varphi_{2n}(x),\qquad \sum_{n=0}^\infty\frac{\lambda_a(z)^n}{\alpha_k}\Psi_{2n}(x),
\label{33}\en
where $\lambda_a(z)=\frac{1}{2}e^{-i\theta}\tanh r$ and $ \lambda_b(z)=-\overline{\lambda_a(z)}$, are all convergent $\forall z=re^{i\theta}\in C_\rho(0)$, where $C_\rho(0)$ is the circle centered in the origin of the complex plane and of radius $$\rho=\min\left[\tanh^{-1}\left(\frac{2\overline\alpha M(\varphi)}{r_\varphi^2}\right),\tanh^{-1}\left(\frac{2\overline\alpha M(\Psi)}{r_\Psi^2}\right)\right].$$

\end{thm}

\begin{proof}
The proof relies upon the following estimates
\beano
 \sum_{n\geq0}\frac{|\tanh(r)|^n}{2^n\alpha_n}\lVert \varphi_{2n}\rVert\leq
 \sum_{n\geq0}\frac{|\tanh(r)|^n}{2^n\alpha_n}A_{\varphi}r_{\varphi}^{2n}M_{2n}(\varphi),\\
 \sum_{n\geq0}\frac{|\tanh(r)|^n}{2^n\alpha_n}\lVert \Psi_{2n}\rVert\leq
 \sum_{n\geq0}\frac{|\tanh(r)|^n}{2^n\alpha_n}A_{\Psi}r_{\Psi}^{2n}M_{2n}(\Psi),
\enano
and from a straightforward determination of the radii of convergence of the latter series.

\end{proof}
The above theorem can be used to estimate the convergence of the bi-squeezed states \eqref{bisquizzi}. For that we take $\alpha_n=\frac{n!}{\sqrt{(2n)!}}$ in
\eqref{33}. With this choice we find $\overline{\alpha}=\frac{1}{2}$, whereas the explicit values of $A_\varphi$, $A_\Psi$, $r_\varphi,r_\Psi$ have to be fixed according to the specific expression of  the states $\varphi_n,\Psi_n$, and of their norms, as the next example shows.

\subsubsection{A case study: the Swanson model}

The non-hermitian  Swanson model arises, in its 1D version,
from the non-hermitian Hamiltonian
$$
H_{\nu}=\frac{1}{2\cos(2\nu)}\left(\hat p_0^2e^{-2i\nu}+\hat q_0^2e^{2i\nu}\right),$$ %
where  $\hat q_0$ and $\hat p_0$ are the self-adjoint position and momentum operators, see Section \ref{mur}, and $\nu$ is a real parameter taking values in $I:=(-\frac{\pi}{4},\frac{\pi}{4})\backslash\{0\}$, see \cite{swan}\footnote{For $\nu=0$ we  recover the harmonic oscillator Hamiltonian}.
Introducing the pair of pseudo bosonic operators defined as
\be
\hat a=\frac{1}{\sqrt{2}}\left(\hat q_0 e^{i\nu}+i\hat p_0e^{-i\nu}\right),\quad \hat b=\frac{1}{\sqrt{2}}\left(\hat q_0 e^{i\nu}-i\hat p_0e^{-i\nu}\right),
\label{231},\en
see \cite{baginbagbook}, they satisfy
\be
[\hat a,\hat b]=\1,\quad \hat a^\dag\neq \hat b,
\label{24}\en
and moreover
$$H_\nu=\frac{1}{\cos(2\nu)}\left(\hat b \hat a+\frac{1}{2}\1\right).$$
As shown in \cite{baginbagbook}, $\D=\SC$, and the biorthonormal families $F_{\varphi}$ and $\F_\Psi$ are defined by the functions
\be
\varphi_n(x)=\frac{N_1}{\sqrt{2^nn!}}H_n(e^{i\nu}x)\textrm{exp}\left\{-\frac{1}{2}e^{2i\nu}x^2\right\},\quad\Psi_n(x)=\frac{N_2}{\sqrt{2^nn!}}H_n(e^{-i\nu}x)\textrm{exp}\left\{-\frac{1}{2}e^{-2i\nu}x^2\right\}
\en
for all $n\geq0$. Here, to guarantee that $\left<\varphi_{ 0},\Psi_0\right>=1$, we take $N_1\bar{N}_2=\frac{e^{-i\nu}}{\sqrt{\pi} }$. In the following we further fix $N_1=1$.

The bi-squeezed states $\tau_z(x),\kappa_z(x)$ in (\ref{bisquizzi}) turn out to be
\begin{eqnarray}
\tau_z(x)&=&e^{\lambda(z)}\sum_{k\geq0}\frac{\lambda_b(z)^k}{k!\sqrt{2^{2k}}}H_{2k}(e^{i\nu}x)\textrm{exp}\left\{-\frac{1}{2}e^{2i\nu}x^2\right\},\\ \kappa_z(x)&=&\frac{e^{i\nu}e^{\lambda(z)}}{\sqrt{\pi}}\sum_{k\geq0}\frac{\lambda_b(z)^k}{k!\sqrt{2^{2k}}}H_{2k}(e^{-i\nu}x)\textrm{exp}\left\{-\frac{1}{2}e^{-2i\nu}x^2\right\}.\label{bisquizzi_swa}
\end{eqnarray}

Using the equality 
\beano
\lVert\varphi_{n}\rVert^2= \sqrt{\frac{\pi}{\cos(2\nu)}}\mathcal{L}_n\left(\frac{1}{\cos(2\nu)}\right),
\enano
where $\mathcal{L}_n$ is Legendre polynomial of degree $n$, \cite{baginbagbook},
we obtain that
\bea
\nonumber\lVert\tau_z\rVert^2=\frac{\pi e^{|\lambda(z)|^2}}{\cos(2\nu)}\sum_{k\geq0}\frac{1}{2^{2k}}\frac{|\tanh(|z|)|^{2k}}{(k!)^2}(2k)!
\mathcal{L}_{2k}\left(\frac{1}{\cos(2\nu)}\right).
\ena
%
With the Laplace-Heine asymptotic formula, \cite{Sze}, pag. 194, Th. 8.21.1,
$$\mathcal{L}_n(x)\approx\frac{1}{\sqrt{2\pi n}}(x^2-1)^{-1/4}\left(x+(x^2-1)^{1/2}\right)^{n+1/2},$$
valid for $n\rightarrow \infty$ and $x\in\mathbb{R}\backslash[-1,1]$,
it is straightforward to prove that $\lVert\tau_z\rVert$ converges in the ball $C_{\rho_\nu}(0)$
with $\rho_\nu=\tanh^{-1}\left(\frac{1}{\cos(2\nu)}+\left(\frac{1}{\cos^2(2\nu)}-1\right)^{1/2}\right)^{-2}$. This follows from Theorem \ref{theo2}, with the following identifications:
\beano
A_{\varphi}&=&\left(\frac{\pi}{\cos(2\nu)}\right)^{1/4}\left(\left(\Cv\right)^2-1\right)^{-1/8}\left(\Cv+\left(\left(\Cv\right)^2-1\right)^{1/2}\right)^{1/4},\\
M_n(\phi)&=&\frac{1}{(2\pi n)^{1/4}},\\
r_{\varphi}&=&\left(\Cv+\left(\left(\Cv\right)^2-1\right)^{1/2}\right)^{1/2},
\enano
where $\alpha_n=\frac{n!}{\sqrt{(2n)!}}$ in   \eqref{33}.
Similar estimates can be repeated for $\kappa_z(x)$.
It follows that the radius of convergence of $\tau_z(x)$ and $\kappa_z(x)$ shrinks to zero for $\nu\rightarrow \pm \frac{\pi}{4}$,
while convergence in the whole complex plane is deduced for $\nu\rightarrow0$, that is when
the Hermiticity of $H_{\nu}$ is recovered. This is, in fact, not surprising: in this limit, in fact, we go back to the standard situation, where the BSS collapse into the single, standard and always well-defined, squeezed state.

It would be interesting to consider a possible extension of these results to other kind of generalized Swanson models, like the one discussed in \cite{postref}, when the mass depends on position. This is part of our future plans.

\section{Bi-squeezed states in a physical system}\label{sect::example}

In  \cite{cohen} the following hermitian Hamiltonian 
$$
H_0=\omega \hat a_0^\dagger \hat a_0+
i\Lambda\left((\hat a_0^\dagger)^2e^{-2i\omega t}-\ha_0^2 e^{2i\omega t}\right),
$$
is introduced, 
in connection with ordinary squeezed states. Here
 $\omega$ and $\Lambda$ are real parameters.  If we replace bosonic with pseudo-bosonic operators, $H_0$ is replaced by
\be
H=\omega  \hat b \hat a+i\Lambda\left({\hat b}^2e^{-2i\omega t}-\hat a^2e^{2i\omega t}\right).
\label{41}\en
Here $\hat a$ and $\hat b$ are any pair of operators satisfying Definition \ref{def21}, while $\omega$ and $\Lambda$ are as above. It is clear that $H$ is not hermitian. As for its physical meaning, let us consider the simple situation in which
	$$
	\hat a = \hat a_0 + \beta,\quad \quad \hat b = \hat a_0^{\dagger} + \gamma,
	$$
	with constants $\beta\ne\gamma$ both real and much smaller than  $\Lambda$, which is much smaller than $\omega$. Then $H$ can be approximated as
	$$H\simeq H_0+H_1$$
	$$H_0=\omega \hat a_0^\dagger \hat a_0+
	i\Lambda\left((\hat a_0^\dagger)^2e^{-2i\omega t}-\ha_0^2 e^{2i\omega t}\right),$$
	$$H_1=\gamma\omega \hat a_0+\beta\omega\hat a_0^\dagger, $$
	neglecting  terms which are quadratic in  $\beta$ and $\gamma$ or depend on $\Lambda\gamma$ and on $\Lambda\beta$. Following \cite{WallsMil}, this new Hamiltonian describes a specific problem in quantum optics: a parametric oscillator ($H_0$) in the presence of cavity losses ($H_1$). As we will see below,  these effects produce a dynamics which can be easily described in terms of pseudo-bosons.


If we introduce the {\em capital operators} $\hat A(t)=\hat a(t)\,e^{i\omega t}$ and $\hat B(t)=\hat b(t)\,e^{-i\omega t}$, and their linear combinations 
$$
\hat X_+(t)=\frac{1}{2}\left(\hat A(t)+\hat B(t)\right), \qquad \hat X_-(t)=\frac{1}{2i}\left(\hat A(t)-\hat B(t)\right),
$$
the Heisenberg equations of motion for these latter can be easily solved, and we find that $\hat X_\pm(t)=\hat X_\pm(0)e^{\pm2\Lambda t}$. Here $\hat X_+(0)=\frac{\hat a+\hat b}{2}$ and $\hat X_-(0)=\frac{\hat a-\hat b}{2i}$. Therefore,
$$
\hat A(t)=\hat a\cosh(2\Lambda t)+\hat b \sinh(2\Lambda t), \qquad \hat B(t)=\hat b\cosh(2\Lambda t)+\hat a \sinh(2\Lambda t).
$$
If we consider the time-dependent number operator $\hat N(t)=\hat B(t)\hat A(t)$, its mean value on $\varphi_0$ turns out to be
$$
\left<\varphi_0, \hat N(t)\varphi_0\right>=\|\varphi_0\|^2\sinh^2(2\Lambda t)+ \left<\varphi_0, \hat b^2\varphi_0
\right>\sinh(2\Lambda t)\cosh(2\Lambda t).
$$
Notice that the second term is not zero, in general, since the matrix element is proportional to  $ \left<\varphi_0, \varphi_2\right>$, and the vectors in $\F_{\varphi}$ are not orthogonal, neither normalized. However, if we replace the mean value of $\hat N(t)$ with the matrix element $\left<\Psi_0, \hat N(t)\varphi_0\right>$, we obtain 
$$
\left<\Psi_0, \hat N(t)\varphi_0\right>=\sinh^2(2\Lambda t),
$$
which is the same result we would get when going back from pseudo-bosons to ordinary bosons, \cite{cohen}. In the same way, to compute the mean value of $\hat N^\dagger(t)$, rather than considering $\left<\Psi_0, \hat N^\dagger(t)\Psi_0\right>$, it is convenient to compute  $\left<\varphi_0, \hat N^\dagger(t)\Psi_0\right>$, which again returns $\sinh^2(2\Lambda t)$, since $\left<\varphi_0, \hat N^\dagger(t)\Psi_0\right>=\overline{\left<\Psi_0, \hat N(t)\Psi_0\right>}$. Hence the following natural questions arises: which kind of matrix elements does really make sense, here? And why? A similar question was discussed in \cite{bagdin1} and \cite{bagdin2} from the point of view of the dynamics of the system, and the analysis was linked to the presence of the same operator $S_\Psi$ which turns out to be useful also here.   In fact, introducing the positive operator $S_\Psi=\sum_n|\Psi_n><\Psi_n|$, see \cite{baginbagbook}, biorthogonality of $\F_{\varphi}$ and $\F_{\Psi}$ implies that $S_\Psi\varphi_n=\Psi_n$. We refer to \cite{baginbagbook} for several mathematical aspects of $S_\Psi$ and of its inverse, including the convergence of the series which define these operators. Here we only want to observe that, for any operator $\hat Q$ on $\Hil$,
$$
\left<\Psi_0, \hat Q\varphi_0\right>=\left<S_\Psi\varphi_0, \hat Q\varphi_0\right>=\left<\varphi_0, \hat Q\varphi_0\right>_{S_\Psi},
$$
where we have introduced the new scalar product $\left<\cdot, \cdot\right>_{S_\Psi}=\left<\cdot, S_\Psi\cdot\right>$, analogously to what we have done at the end of Section \ref{sec::prel},  where the role of $S_\Psi$ was played by $\eta$. Hence it is easy to understand what we are doing when computing matrix elements as the one in $\left<\Psi_0, \hat Q\varphi_0\right> $: we are only computing the mean value of $\hat Q$ on $\varphi_0$, but with respect to a different scalar product. This is not surprising. On the contrary, it is in fact a typical aspect of PT-quantum mechanical systems, \cite{BB98}-\cite{baginbagbook}.

Following \cite{cohen}, we can rewrite the non-hermitian version of the electric field
$$
E(x,t)=i\left(\hat a(t)e^{ikx}-\hat b(t)e^{-ikx}\right),
$$
as follows:
$$
E(x,t)=-2\left(\hat X_+(0)e^{2\Lambda t}\sin(kx-\omega t)+\hat X_-(0)e^{-2\Lambda t}\cos(kx-\omega t)\right),
$$
which shows that one component of $E(x,t)$ is amplified, and the other is damped. If we compute the matrix elements for $\hat X_\pm(t)$ and their squares, and we introduce a sort of deformed variance for the operator $\hat G$ as 
$$
(\delta \hat G)^2=\left<\Psi_0,\hat G^2\varphi_0\right>-\left<\Psi_0, \hat G\varphi_0\right>^2,
$$
we easily find that $\delta X_+(t)\delta X_-(t)=\frac{1}{2}$, for all $t$. Then, in view of our previous comment, this suggests that the {\em good} scalar product to adopt, at least if we are interested in saturating the Heisenberg inequality, is the $\left<., .\right>_{S_\Psi}$ one. Incidentally we observe that $\delta \hat G$ is obtained similarly to $(\Delta_{\eta}\hat G)_\chi$ in  Section \ref{mur}, and in this perspective the comments given there still hold here. What is interesting for us is that the exponential of $H$, when $\omega=0$, can be identified with the operator $\Sc(z)$, with $z=-2\Lambda$. Of course, the exponential of $H^\dagger$ is nothing but $\mathcal{T}(z)$, with the same identification. Then we conclude that our generalized squeezing operators can be related to some quadratic Hamiltonian deduced easily with a simple deformation of bosonic operators, which must be replaced with their pseudo-bosonic counterparts.

\section{Conclusions and possible developments}\label{sect::concl}
In this paper we have introduced a new class of states, the BSS, and we have deduced some of their properties. We have found three equivalent definitions for the regular BSS, while the non regular ones are conveniently defined in form of series for which a rather mild condition of convergence has been proposed. { Some examples  of BSS in concrete  physical models described by non hermitian Hamiltonians have been discussed and analysed.
We have also shown that BSS, when considered together with suitable metric operators, are able to saturate the Heisenberg uncertainty inequality.}

We plan to analyze in more details the role of these vectors in the context of PT-quantum mechanics, and to look for more properties and for more applications. The dynamics of these states obviously also deserve attention. An interesting question, for instance, is: does the time evolution maps BSS into (possibly different) BSS? Another intriguing aspect is whether it is possible to construct some experimental settings in which they can be observed. These are some of the aspects which we plan to consider next in our analysis.

\section*{Acknowledgements}
The authors acknowledge partial support from Palermo University. F.B. and F.G. also acknowledge partial support from G.N.F.M. of the I.N.d.A.M.

\section*{Computational solution}

This paper does not contain any computational solution.

\section*{Ethics statement}

This work did not involve any active collection of human data.

\section*{Data accessibility statement}

This work does not have any experimental data.

\section*{Competing interests statement}

We have no competing interests.

\section*{Authors' contributions}

All authors equally contributed to the preparation of the paper, and gave final approval for publication.

\section*{Funding}

This work was partly supported by G.N.F.M. and G.N.A.M.P.A.-INdAM and by the University of Palermo.

%
%
%
%
%
%



\vskip2pc


\begin{thebibliography}{9}

	
\bibitem{effH1}	G. L. Celardo and L. Kaplan, {\em Superradiance transition in one-dimensional nanostructures: An effective non-Hermitian Hamiltonian formalism}, Phys. Rev. B, {\bf 79}, 155108 (2009)


\bibitem{effH2} G. G. Giusteri,  F. Mattiotti  and G. L. Celardo,  {\em Non-Hermitian Hamiltonian approach to quantum transport in disordered networks with sinks: validity and effectiveness}, Phys. Rev. B, {\bf 91}, 094301 (2015)

	
	
\bibitem{effH3} D. V. Savin, V. V. Sokolov and H.-J. Sommers, {\em Is the concept of the non-Hermitian effective Hamiltonian relevant in the case of potential scattering?}, Phys. Rev. E, {\bf 67}, 026215, (2003)
	
	
\bibitem{BB98}C. M. Bender and S. Boettcher, \emph{Real Spectra of Non-Hermitian Hamiltonian Having ${\mathcal PT}$ Symmetry}, Phys. Rev. Lett. \textbf{80}, 5243 (1998).
	
	
\bibitem{dapro} J. da Provid$\hat e$ncia, N. Bebiano, J.P. da Provid$\hat e$ncia, {\em Non hermitian operators with real spectrum in quantum mechanics}, ELA, {\bf 21}, 98-109 (2010)
	
	
\bibitem {MosPra09} A. Mostafazadeh, \emph{Non-Hermitian Hamiltonians with a real spectrum and their physical applications}, Pramana-J. Phys. \textbf{73}, 269-277 (2009).
	
\bibitem{BGV15} F. Bagarello, F. Gargano, D. Volpe {\em $\mathcal{D}$-Deformed Harmonic Oscillator}, Int. J. Theor. Phys., {\bf 54}(11),4110-4123 (2015)
	
	
\bibitem{bagdin1} F. Bagarello, {\em Some results on the dynamics and  transition probabilities for non self-adjoint hamiltonians},  Ann. of Phys., {\bf 356}, 171-184 (2015)



\bibitem{bagdin2} F. Bagarello, {\em Transition probabilities for non self-adjoint Hamiltonians in infinite dimensional Hilbert spaces},
Ann. of Phys., {\bf 362}, 424-435 (2015)


\bibitem{bagdin3} F. Bagarello, {\em Non self-adjoint Hamiltonians with complex eigenvalues},  J. Phys. A, {\bf 49}, 215304 (2016)

	
\bibitem{bagproc} F. Bagarello, R. Passante, C. Trapani, {\em Non-Hermitian Hamiltonians in Quantum Physics;
	Selected Contributions from the 15th International Conference on Non-Hermitian
	Hamiltonians in Quantum Physics, Palermo, Italy, 18-23 May 2015}, Springer (2016)
	
	
\bibitem{baginbagbook} F. Bagarello, {\em Deformed canonical (anti-)commutation relations and non hermitian Hamiltonians}, in {Non-selfadjoint operators in quantum physics: Mathematical aspects}, F. Bagarello, J. P. Gazeau, F. H. Szafraniec and M. Znojil Eds., John Wiley and Sons Eds, Hoboken, New Jersey (2015)

\bibitem{BNLPbs} F. Bagarello, {\em Non linear pseudo-bosons}, J. Math. Phys.,  {\bf 52}, 063521, (2011)

	
\bibitem{BZn1} F. Bagarello, M. Znojil, Non Linear pseudo-bosons versus hidden Hermiticity, J. Phys. A, {\bf 44}, 415305 (2011)

\bibitem{BZn2} F. Bagarello, M. Znojil, Non Linear pseudo-bosons versus hidden Hermiticity. II: the case of unbounded operator, J. Phys. A, {\bf 45},115311 (2012)

\bibitem{tri} D.A. Trifonov, {Pseudo-boson coherent and Fock states}, arXiv: quant-ph/0902.3744, Proceedings of the 9th International Workshop on Complex Structures, Integrability and Vector Fields, Sofia, August 2008, 241-250


\bibitem{bagpb1} F. Bagarello, {\em Pseudo-bosons, Riesz bases and coherent states}, J. Math. Phys., {\bf 50}, 023531 (2010) (10pg)

	
\bibitem{BagQUO} F. Bagarello, 
{\em Deformed quons and bi-coherent states}
Proceedings of the Royal Society A: Mathematical, Physical and Engineering Sciences, 473(2200):20170049 (2017).

\bibitem{bialo2017} F. Bagarello, F. Gargano, S. Spagnolo, {\em Two-dimensional non commutative  Swanson model and its bicoherent states},  Proceedings of the WGMP Conferences,
2017, Bialowieza,  Polonia, in press
	
\bibitem{gazeaubook} J.P. Gazeau, {\em Coherent States in Quantum Physics},   WILEY-VCH verlag GmbH and Co. KGaA, Weinheim, (2009)
	
	
\bibitem{didier} M. Combescure,  R. Didier, {\em Coherent States and Applications in Mathematical Physics},   Springer, (2012)
	
\bibitem{aag}S. T. Ali, J. P. Antoine and J.P. Gazeau,	{\em Coherent States, Wavelets, and Their Generalizations}, Springer 2014
	
	

\bibitem{BarRad}  S.M. Barnett, P.M. Radmore, {\em Methods in Theoretical Quantum Optics}, Clarendon Press-Oxford Science Publications, 2002


\bibitem{ScullyZub} M.O. Scully, M.S. Zubairy, {\em Quantum Optics}, Cambridge University Press, 2002

\bibitem{WallsMil} D.F. Walls, G. J. Milburn, {\em Quantum Optics}, Springer, 2008

\bibitem{tri2} O. Cherbal, M. Drir, M. Maamache , D. A. Trifonov, {\em Fermionic coherent states for pseudo-Hermitian two-level systems},
J. Phys. A, {\bf 40}, 1835-1844, (2007)

\bibitem{maleki} Y. Maleki, {\em Para-Grassmannian Coherent and Squeezed States for Pseudo-Hermitian $q$-Oscillator and their Entanglement},  
SIGMA {\bf 7}, 084, 2011

\bibitem{BGST2017} F. Bagarello, F. Gargano, S. Spagnolo, S.Triolo, {\em Coordinate representation for non-Hermitian position and momentum operators},  Proceedings of the Royal Society A, \textbf{473} (2205), 20170434, 2017


\bibitem{Nori2011} C. M. Wilson, G. Johansson, A. Pourkabirian, M. Simoen, J. R. Johansson,	T. Duty, F. Nori	and P. Delsing, {\em Observation of the dynamical Casimir effect in a superconducting circuit}, Nature \textbf{479}, 376-379 (2011)



\bibitem{intop} Kuru S., Tegmen A., Vercin A., {\em Intertwined isospectral potentials in an arbitrary dimension},
J. Math. Phys, {\bf 42}, No. 8, 3344-3360, (2001); Kuru S.,
Demircioglu B., Onder M., Vercin A., {\em Two families of
	superintegrable and isospectral potentials in two dimensions}, J.
Math. Phys, {\bf 43}, No. 5, 2133-2150, (2002); Samani K. A., Zarei
M., {\em Intertwined hamiltonians in two-dimensional curved spaces},
Ann. of Phys., {\bf 316}, 466-482, (2005); N. Aizawa, V. K. Dobrev, {\em Intertwining Operator Realization of Non-Relativistic Holography}, Nucl. Phys. B {\bf 828}, 581-593 (2010); B. Midya, B. Roy, R. Roychoudhury, {\em Position Dependent Mass Schroedinger Equation and Isospectral Potentials : Intertwining Operator approach}, J.  Math. Phys., {\bf 51}, 022109 (2010); A. L. Lisok, A. V. Shapovalov, A. Yu. Trifonov, {\em Symmetry and Intertwining Operators for the Nonlocal Gross-Pitaevskii Equation}, SIGMA {\bf 9}, 066, 21 pages (2013)

\bibitem{maa} O. Cherbal, M. Maamache,{\em
Time-dependent pseudofermionic systems and coherent states},
J. Math. Phys  \textbf{57}, 022102, (2016)

\bibitem{swan} M. S. Swanson, {\em Transition elements for a non-Hermitian quadratic Hamiltonian},
J. Math. Phys., {\bf 45}, 585, (2004)


\bibitem{Sze} G.Szego, {\em Orthogonal polynomials}, American Mathematical Society, Colloquium Publications, Vol. 23, 4th
ed., Amer. Math. Soc., Providence, R.I., 1975.

\bibitem{postref}  B. Midya, P P Dube, R. Roychoudhury, {\em Non-isospectrality of the generalized Swanson Hamiltonian and harmonic oscillator}, J. Phys. A: Math. Theor., {\bf 44}, 062001 (2012)



%
%
%
%
%


\bibitem{cohen} C. Cohen-Tannoudji, J. Dupont-Roc, G. Grynberg, {\em Photons and Atoms, Introduction to Quantum Electrodynamics}, John Wiley and Sons, 1989


\end{thebibliography}
\end{document}